\documentclass[letterpaper, 11pt, thm-restate]{article}

\usepackage[margin=1in]{geometry}
\usepackage{footmisc}
\usepackage{enumitem}
\usepackage[title]{appendix}
\usepackage{xcolor}
\usepackage{array}
\usepackage{amsthm}
\usepackage{amsmath}
\numberwithin{equation}{section}

\usepackage{amssymb}
\usepackage{amsfonts}
\usepackage{graphicx}
\usepackage{thm-restate}
\usepackage[colorlinks = true]{hyperref}
\hypersetup{
    linkcolor=black,
    citecolor=blue,
    urlcolor=blue}

\usepackage{mathtools}
\usepackage{mathrsfs}

\usepackage{tikz}
\usetikzlibrary{positioning,calc,arrows,automata,shapes,shapes.geometric,decorations.pathmorphing}
\usetikzlibrary{shapes}

\usepackage{titlesec}
\titleformat{\subsection}[runin]
       {\normalfont\bfseries}
       {\thesubsection}
       {0.5em}
       {}
       [.]

\newcommand{\Z}{\mathbb{Z}}
\newcommand{\N}{\mathbb{N}}
\newcommand{\Q}{\mathbb{Q}}

\newcommand{\F}{\mathbb{F}}

\newcommand{\Zpe}{\mathbb{Z}_{/p^e}}

\newcommand{\ZT}{\mathbb{Z}_{/T}}

\newcommand{\mM}{\mathcal{M}}

\newcommand{\mI}{\mathcal{I}}

\newcommand{\mS}{\mathcal{S}}

\newcommand{\mA}{\mathcal{A}}
\newcommand{\mZ}{\mathfrak{Z}}
\newcommand{\mX}{\mathcal{X}}

\newcommand{\mR}{\mathcal{R}}

\newcommand{\ba}{\boldsymbol{a}}
\newcommand{\bb}{\boldsymbol{b}}

\newcommand{\bd}{\boldsymbol{d}}

\newcommand{\tmA}{\widetilde{\mathcal{A}}}

\newcommand{\bh}{\boldsymbol{h}}

\newcommand{\gen}[1]{\langle {#1} \rangle}

\newtheorem{thrm}{Theorem}[section]
\newtheorem*{thrm*}{Theorem}
\newtheorem{lem}[thrm]{Lemma}
\newtheorem{obs}[thrm]{Observation}

\newtheorem{cor}[thrm]{Corollary}
\theoremstyle{definition}
\newtheorem{defn}[thrm]{Definition}
\newtheorem*{defn*}{Definition}
\theoremstyle{definition}

\theoremstyle{definition}

\title{The Skolem Problem in rings of positive characteristic}
\author{Ruiwen Dong\footnote{Magdalen College, University of Oxford, United Kingdom, email: ruiwen.dong@magd.ox.ac.uk} \and Doron Shafrir\footnote{Department of Mathematics, Ben Gurion University of the Negev, Be’er Sheva, Israel}}
\date{ }

\begin{document}
\maketitle

\begin{abstract}
    We show that the Skolem Problem is decidable in finitely generated commutative rings of positive characteristic.
    More precisely, we show that there exists an algorithm which, given a finite presentation of a (unitary) commutative ring $\mathcal{R} = \mathbb{Z}_{/T}[X_1, \ldots, X_n]/I$ of characteristic $T > 0$, and a linear recurrence sequence $(\gamma_n)_{n \in \mathbb{N}} \in \mathcal{R}^{\mathbb{N}}$, determines whether $(\gamma_n)_{n \in \mathbb{N}}$ contains a zero term.
    Our proof is based on two recent results: Dong and Shafrir (2026) on the solution set of S-unit equations over $p^e$-torsion modules, and Karimov, Luca, Nieuwveld, Ouaknine, and Worrell (2025) on solving linear equations over powers of two multiplicatively independent numbers.
    Our result implies, moreover, that the zero set of a linear recurrence sequence over a ring of characteristic $T = p_1^{e_1} \cdots p_k^{e_k}$ is effectively a finite union of $p_i$-normal sets in the sense of Derksen (2007). 
\end{abstract}

\section{Introduction and Main Results}

The Skolem Problem in a commutative ring $\mR$ is the problem of determining whether a given linear recurrence sequence $(\gamma_n)_{n \in \N}$ over $\mR$ contains a zero term.
The Skolem Problem is deeply related to central questions in program verification~\cite{ouaknine2015linear, almagor2021deciding}, control theory~\cite{blondel2000survey, fijalkow2019decidability}, dynamical systems~\cite{kannan1986polynomial, agrawal2015approximate}, and number theory~\cite{bell2006generalised, luca2023skolem}.
Despite its simple formulation, decidability of the Skolem Problem over $\Z$ (and equivalently, $\Q$) is a major open problem in mathematics and computer science: see~\cite{mahler1935arithmetische, lech1953note, Tijdeman1984, blondel2002presence, lipton2022skolem, luca2025large} for a selected list of partial results and recent progress.
Indeed, the celebrated Skolem-Mahler-Lech theorem~\cite{skolem1934verfahren} shows that the zero set of a linear recurrence sequence over $\Z$ is a union of a finite set and finitely many arithmetic progressions.
However, the Skolem-Mahler-Lech theorem is not \emph{effective}, meaning there is currently no known computable bound on this finite set, and hence the \emph{decidability} of the Skolem Problem over the ring $\Z$ remains open.

In contrast, decidability of the Skolem Problem in fields of prime characteristic (such as $\F_p(X)$) was shown by Derksen~\cite{derksen2007skolem}, who gave an \emph{effective} description of the zero set as a \emph{$p$-normal set}.
In this paper we continue Derksen's work in this direction by showing decidability of the Skolem Problem in \emph{rings} of arbitrary positive characteristics.
In contrast to fields, the characteristic of such a ring is not necessarily prime (e.g., the polynomial ring $\Z_{/6}[X]$ has characteristic $6$).
Thus a difficulty is to describe the interaction of different prime divisors of the characteristic.

Let $\mR$ be a commutative ring with unity $1$.
The \emph{characteristic} of the ring $\mR$ is defined as the smallest positive integer $T$ such that $T \cdot 1 = 0$ in $\mR$.
If no such positive integer exists, the ring is said to have characteristic zero.
A sequence $\gamma = (\gamma_0, \gamma_1, \gamma_2, \ldots) \in \mR^{\N}$ is called an \emph{linear recurrence sequence} if there exist $d \geq 1$ and $a_1, \ldots, a_d \in \mR$ with $a_d \neq 0$, such that
\begin{equation}\label{eq:recurrence}
    \gamma_{n} = a_1 \gamma_{n-1} + \cdots + a_d \gamma_{n-d}
\end{equation}
for all $n \geq d$.
Note that the sequence $\gamma$ is uniquely determined by the recurrence relation~\eqref{eq:recurrence} as well as its $d$ initial terms $\gamma_0, \ldots, \gamma_{d-1}$.
Since every term in $\gamma$ is an element in the subring generated by $a_1, \ldots, a_d, \gamma_0, \ldots, \gamma_{d-1}$, by restricting to this subring we can suppose $\mR$ to be finitely generated.
The main result of this paper is:
\begin{thrm}\label{thm:Skolem}
    Let $\mR$ be a finitely generated commutative (unitary) ring of characteristic $T > 0$.
    Given a linear recurrence sequence $\gamma \in \mR^{\N}$, it is decidable whether $\gamma$ contains a zero.
\end{thrm}

We say a few words about how the ring $\mR$ is represented.
Denote by $\ZT$ the quotient ring $\Z/T\Z = \{0, 1, \ldots, T-1\}$ and let $r_1, \ldots, r_N$ be the generators of $\mR$.
Since $T \cdot 1 = 0$ in $\mR$, we have a surjective homomorphism $\varphi$ from the polynomial ring $\ZT[X_1, \ldots, X_N]$ to $\mR$, defined by $X_i \mapsto r_i, i = 1, \ldots, N$.
Therefore $\mR$ can be written as a quotient $\ZT[X_1, \ldots, X_N]/I$, with $I = \ker(\varphi)$.
This ideal $I$ is finitely generated since $\ZT[X_1, \ldots, X_N]$ is Noetherian.
Therefore, throughout this paper, the ring $\mR$ is always represented as a quotient $\ZT[X_1, \ldots, X_N]/I$, where the generators of $I$ are explicitly given.

Theorem~\ref{thm:Skolem} for rings of \emph{prime} characteristics can be deduced from the case of fields following Derksen's result.
By~\cite[Section~9]{derksen2007skolem}, given any linear recurrence sequence $\gamma$ over a finitely generated ring of prime characteristic $p$, its zero set $\mZ(\gamma) \coloneqq \{n \in \N \mid \gamma_n = 0\}$ is effectively \emph{$p$-normal} (see Section~\ref{sec:prelim} for a formal definition).
For example, the linear recurrence sequence $\gamma$ over the polynomial ring $\F_p[X]$, defined by $\gamma_n = (X+1)^n - X^n - 1$, admits the zero set $\{n = p^k \mid k \in \N\}$~\cite[Example~1.3]{derksen2007skolem}.

\paragraph*{Overview of main contributions and proof structure.}

More generally, when considering a linear recurrence sequence $\gamma$ over a ring of \emph{non-prime} characteristic $T = p_1^{e_1} \cdots p_k^{e_k}$, we can use the Chinese Remainder Theorem to decompose the zero set $\mZ(\gamma)$ as an intersection $\mZ(\alpha_1) \cap \cdots \cap \mZ(\alpha_k)$ of zero sets, where each $\alpha_i$ is a linear recurrence sequence over a ring of characteristic $p_i^{e_i}$.
Then, in order to decide whether $\gamma$ contains a zero (i.e., whether $\mZ(\gamma)$ is non-empty), one is immediately confronted with two major obstacles.

The first obstacle is that Derksen's result does not apply to rings of prime-power characteristics $p^e, e \geq 2$, due to its reliance on the \emph{Frobenius homomorphism}.
To overcome this, in this paper we will extend Derksen's result from prime characteristic to prime-power characteristics $p^e$.
Namely, we will show the following.

\begin{restatable}{prop}{propSkolemprimepower}\label{prop:Skolem_primepower}
    Let $p$ be a prime number, $e$ be a positive integer, and $\mA$ be a finitely generated commutative unitary ring of characteristic $p^e$.
    For any given linear recurrence sequence $\alpha \in \mA^{\N}$, its zero set $\mZ(\alpha) \coloneqq \{n \in \N \mid \alpha_n = 0\}$ is effectively $p$-normal.
\end{restatable}

The formal definition of \emph{$p$-normal sets} will be given in Section~\ref{sec:prelim} (Definition~\ref{def:pnormal}).
An example of $p$-normal sets is $\{1 + 5 \cdot p^{2a} + p^{2b} \mid a, b \in \N\}$.
In particular, these are special forms of \emph{$p$-automatic sets}~\cite{wolper2000construction}, and is a generalization of linear combinations of powers of $p$.
Our key argument to proving Proposition~\ref{prop:Skolem_primepower} is a deep result of Dong and Shafrir~\cite[Theorem~1.3]{dong2025s}, who showed that the solution set to an \emph{S-unit equation} over a $p^e$-torsion module is effectively $p$-normal.


The second obstacle is to compute the intersection $\mZ(\alpha_1) \cap \cdots \cap \mZ(\alpha_k)$ of the zero sets.
Due to Proposition~\ref{prop:Skolem_primepower}, this boils down to computing the intersection of $p_i$-normal sets for \emph{different} primes $p_1, \ldots, p_k$.
In general, it is unknown whether one can effectively compute the intersection of \emph{$p_i$-automatic} sets for different $p_i$'s~\cite{Hieronymi2022, albayrak2023quantitative}.
However, we will show that this is possible in the special case of $p_i$-normal sets.
In particular, we will prove that every \emph{intersection} of $p_i$-normal sets can actually be written as a \emph{union} of $p_i$-normal sets:
\begin{restatable}{prop}{propintersectionisunion}\label{prop:intersection_is_union}
    Let $p_1, \ldots, p_k \in \Z_{>0}$ be multiplicatively independent positive integers.
    For $i = 1, \ldots, k$, let $S_i$ be a $p_i$-normal subset of $\N$.
    Then the intersection $S_1 \cap \cdots \cap S_k$ is effectively equal to a union $T_1 \cup \cdots \cup T_k$, where each $T_i$ is a $p_i$-normal subset of $\N$.
    In particular, it is decidable whether the intersection $S_1 \cap \cdots \cap S_k$ is empty.
\end{restatable}

Our main tool for proving Proposition~\ref{prop:intersection_is_union} comes from a recent paper by Karimov, Luca, Nieuwveld, Ouaknine, and Worrell~\cite{karimov2025decidability}, which gave an effective procedure to decide the existential fragment of Presburger arithmetic extended with \emph{two} power predicates.
For example, this allows an effective resolution of equations of the form $1 + 5 \cdot p_1^{2a} + p_1^{2b} = 3 + p_2^c$, over the variables $a, b, c \in \N$.
However, the effective procedure in~\cite{karimov2025decidability} does not extend to \emph{more than two} power predicates.
Therefore, additional insights are needed to compute the intersection of $p_i$-normal sets for $k \geq 3$ distinct primes $p_1, \ldots, p_k$.

Combining Proposition~\ref{prop:Skolem_primepower} and~\ref{prop:intersection_is_union}, one also obtains an effective characterization of the zero set of a linear recurrence sequence over a ring $\mR$ of characteristic $T = p_1^{e_1} \cdots p_k^{e_k}$: namely it is effectively a finite union of $p_i$-normal subset of $\N$ for $i = 1, \ldots, k$.

\section{Preliminaries}\label{sec:prelim}

\subsection{$p$-Normal Sets}
We define $p$-normal subsets of $\Z^n$, $\Z$, and $\N$, following the convention established in~\cite{derksen2007skolem} and~\cite{derksen2015linear}.

\begin{defn}[{reformulation of~\cite{derksen2015linear}}]\label{def:pnormal}
    A set $S \subseteq \Z^n$ is called \emph{elementary $p$-nested} if it is of the form
    \begin{equation}\label{eq:pnested_Zn}
    \{\ba_0 + p^{\ell k_1} \ba_1 + \cdots + p^{\ell k_r} \ba_r \mid k_1, k_2, \ldots, k_r \in \N\},
    \end{equation}
    where $\ell \geq 1$ and $\ba_0, \ba_1, \ldots, \ba_r \in \Q^n$.
    A singleton is by definition elementary $p$-nested with $r = 0$.
    Note that the entries of $\ba_i$'s do not have to be integers: for example, the set $\left\{\frac{1}{2} + 3^k \cdot \frac{1}{2}\;\middle|\; k \in \N\right\} \subseteq \Z$ is elementary $p$-nested.
    
    A subset $S$ of $\Z^n$ is called \emph{$p$-succinct}, if it is of the form $H + D \coloneqq \{\bh + \bd \mid \bh \in H, \bd \in D\}$, where $H$ is a subgroup of $\Z^n$ and $D$ is elementary $p$-nested.
    For example, the set $\{(x, x+p^{k}) \mid x \in \Z, k \in \N\} \subseteq \Z^2$ is $p$-succinct since it is the sum of the subgroup $H = \{(x, x) \mid x \in \Z\} \leq \Z^2$ and the elementary $p$-nested set $D = \{(0, 1) \cdot p^{k} \mid k \in \N\}$.
    
    A subset $S$ of $\Z^n$ is called \emph{$p$-normal} if it is a finite union of $p$-succinct sets.
    The empty set is by definition $p$-normal.
    A set is called \emph{effectively} $p$-normal if all its defining coefficients can be effectively computed.
\end{defn}

Specializing Definition~\ref{def:pnormal} at $n = 1$, we obtain a simpler characterization of $p$-normal subsets of $\Z$.
Indeed, a $p$-succinct subset $S = H + D \subseteq \Z$ is either elementary $p$-nested (if the subgroup $H$ is trivial), or it is a finite union of cosets $(a \Z + b_1) \cup \cdots \cup (a \Z + b_k)$ with $a \geq 1$ (if the subgroup $H$ is of the form $a \Z$).
Therefore, $p$-normal subsets of $\Z$ can be equivalently defined as follows.

\begin{obs}\label{obs:normal_Z}
    A subset $S$ of $\Z$ is $p$-normal, if and only if it is a finite union of elementary $p$-nested sets and sets of the form $a \Z + b$, $a \geq 1, b \in \Z$.
\end{obs}

For subsets of $\N$, the definition is slightly more technical:
\begin{defn}
    A subset $S$ of $\N$ is called \emph{$p$-normal}, if there exists $n \geq 0$, a finite set $F \subseteq \{0, 1, \ldots, n-1\}$, and a $p$-normal subset $S'$ of $\Z$, such that
    \[
    S = \left(S' \cap [n, \infty) \right) \cup F.
    \]
\end{defn}

In other words, a set $S \subseteq \N$ is $p$-normal if it \emph{ultimately} coincides with a $p$-normal subset of $\Z$.
It is clear from definition that $p$-normal sets (of $\Z^n$, $\Z$, and $\N$) are closed under finite union and affine transformations.
The following are some other closure properties for $p$-normal sets:
\begin{lem}[{\cite[Lemma~9.5]{derksen2007skolem}}]\label{lem:intersection_normal_N}
    Let $S_1, S_2$ be $p$-normal subsets of $\N$ (or $\Z$). Then the intersection $S_1 \cap S_2$ is effectively $p$-normal.
\end{lem}

\begin{lem}[{Special case of~\cite[Proposition~2.1]{derksen2015linear}}]\label{lem:nestedZn_to_Z}
    Let $\bb \in \Z^{n}$, $B$ be a subgroup of $\Z^n$ and $U \subseteq \Z^{n}$ be an elementary $p$-nested set.
    Then, the set of integers $\{z \in \Z \mid z \cdot \bb \in B+U\}$ is effectively $p$-normal.
\end{lem}

From Lemma~\ref{lem:nestedZn_to_Z}, we immediately obtain:
\begin{cor}\label{cor:normalZn_to_Z}
    Let $S \subseteq \Z^{n}$ be a $p$-normal set and let $\bb \in \Z^{n}$.
    Then, the set of integers $\{z \in \Z \mid z \cdot \bb \in S\}$ is effectively $p$-normal.
\end{cor}
\begin{proof}
    Since $p$-normal subsets of $\Z$ are closed under finite union, it suffices to prove the case where $S$ is $p$-succinct, which corresponds exactly to Lemma~\ref{lem:nestedZn_to_Z}.
\end{proof}

\subsection{Commutative Algebra}
All rings considered in this paper are commutative and unitary.
We recall some standard definitions from commutative algebra~\cite{eisenbud2013commutative}.

\begin{defn}\label{def:commalg}
    Let $R$ be an effective commutative Noetherian ring.
    \begin{enumerate}[nosep, label=(\roman*)]
        \item An element $r \in R$ is called a \emph{zero-divisor} if there exists a non-zero $x$ in $R$ such that $rx = 0$. For example, $3$ is a zero-divisor of the ring $\Z/12\Z$.
        \item An element $r \in R$ is called \emph{nilpotent} if there exists $n \in \N$ such that $r^n = 0$. For example, $6$ is nilpotent in $\Z/12\Z$.
        \item A proper ideal $I \subsetneq R$ is called \emph{prime} if for every $a, b \in R$ such that $ab \in I$, either $a \in I$ or $b \in I$.
        Equivalently, $I$ is prime if the ring $R/I$ has no zero-divisor other than zero.
        \item A proper ideal $I \subsetneq R$ is called \emph{primary} if for every $a, b \in R$ such that $ab \in I$, either $a \in I$ or some power of $b$ is in $I$.
        Equivalently, $I$ is primary if every zero-divisor of $R/I$ is nilpotent.
        \item Let $I$ be an ideal in $R$. Then $I$ can always be written as a finite union of primary ideals: $I = P_1 \cap \cdots \cap P_n$. Such a finite union is called a \emph{primary decomposition} of $I$.
        Given a finite set of generators of $I$, a primary decomposition of $I$ can be effectively computed~\cite{rutman1992grobner}.
        \item Let $S \subset R$ be a set of non zero-divisors of $R$ and denote by $\widetilde{S} \coloneqq \{s_1 s_2 \cdots s_k \mid k \in \N, s_1, \ldots, s_k \in S\}$ its multiplicative closure. 
        The \emph{localization} $S^{-1}R$ is defined as the ring $\left\{\frac{r}{s} \;\middle|\; r \in R, s \in \widetilde{S}\right\}$, with $\frac{r_1}{s_1} = \frac{r_2}{s_2}$ if and only if $s_1 r_2 - s_2 r_1$ is a zero-divisor.
        This ring is endowed with the operations $\frac{r_1}{s_1} + \frac{r_2}{s_2} = \frac{r_1 s_2 + r_2 s_1}{s_1 s_2}$ and $\frac{r_1}{s_1} \cdot \frac{r_2}{s_2} = \frac{r_1 r_2}{s_1 s_2}$.
        In particular, $S^{-1}R$ contains $R$ as a subring $\left\{\frac{r}{1} \;\middle|\; r \in R\right\}$.
        For example, the localization $\{2\}^{-1} \Z$ is the ring of rational numbers whose denominators are powers of $2$.
        Note that in order for the canonical map $R \rightarrow S^{-1}R$ to be injective, the set $S$ must not contain zero-divisors.
        \item An \emph{$R$-module} is defined as an abelian group $(\mM, +)$ along with an operation $\cdot \;\colon R \times \mM \rightarrow \mM$, satisfying $r \cdot (m+m') = r \cdot m + r \cdot m'$, $(r + s) \cdot m = r \cdot m + s \cdot m$, $rs \cdot m = r\cdot (s \cdot m)$ and $1 \cdot m = m$.
        For example, an ideal of $R$ is an $R$-submodule of $R$. 
        For any $d \in \N$, $R^d$ is an $R$-module by $s \cdot (r_1, \ldots, r_d) = (sr_1, \ldots, sr_d)$.
        \item Given two $R$-modules $\mM, \mM'$ such that $\mM \supseteq \mM'$, define the quotient $\mM/\mM' \coloneqq \{\overline{m} \mid m \in \mM\}$ where $\overline{m_1} = \overline{m_2}$ if and only if $m_1 - m_2 \in \mM'$.
        This quotient is also an $R$-module.
        An $R$-module is called \emph{finitely presented} if it can be written as a quotient $R^d/\gen{v_1, \ldots, v_k}$ for some $d \in \N$ and some $v_1, \ldots, v_k \in R^d$. 
        Here, $\gen{v_1, \ldots, v_k}$ denotes the $R$-submodule generated by the elements $v_1, \ldots, v_k$.
    \end{enumerate}
\end{defn}

For $T > 0$, denote by $\ZT[X_1^{\pm}, \ldots, X_N^{\pm}]$ the Laurent polynomial ring over $\ZT$ with $n$ variables: this is the set of polynomials over the variables $X_1, X_1^{-1}, \ldots, X_N, X_N^{-1}$, with coefficients in $\ZT$, such that $X_i X_i^{-1} = 1$ for all $i$.
The following deep theorem by Dong and Shafrir is key to proving our main result in case of prime-power characteristics.

\begin{restatable}[{\cite[Theorem~1.3]{dong2025s}}]{thrm}{thmprimepower}\label{thm:primepower}
    Let $p$ be a prime number and $e$ be a positive integer.
    Let $\mM$ be a finitely presented module over the Laurent polynomial ring $\Zpe[X_1^{\pm}, \ldots, X_N^{\pm}]$, and let $m_0, m_1, \ldots, m_K$ $\in \mM$.
    Then the set of solutions $(z_{11}, \ldots, z_{KN}) \in \Z^{KN}$ to the equation
    \begin{equation}\label{eq:Sunitthm}
        X_1^{z_{11}} X_2^{z_{12}} \cdots X_N^{z_{1N}} \cdot m_1 + \cdots + X_1^{z_{K1}} X_2^{z_{K2}} \cdots X_N^{z_{KN}} \cdot m_K = m_0
    \end{equation}
    is effectively $p$-normal.
\end{restatable}

Finally, we say a few words about the relation between linear recurrence sequences and formal power series.
Let $R$ be a commutative ring.
A (univariate) \emph{formal power series} over $R$ is a formal sum $\sum_{n=0}^{\infty} a_n Y^n = a_0 + a_1 Y + a_2 Y^2 + \cdots$, where each $a_n$ is an element in $R$.
The set of all formal power series over $R$ forms a ring and is denoted by $R[[Y]]$.

Given a linear recurrence sequence $\gamma = (\gamma_0, \gamma_1, \gamma_2, \ldots) \in R^{\N}$ satisfying the recursion
\[
    \forall n \geq d, \quad \gamma_{n} = a_1 \gamma_{n-1} + \cdots + a_d \gamma_{n-d},
\]
the \emph{characteristic polynomial} of $\gamma$ is defined as
\[
f(Y) \coloneqq Y^d - a_1 Y^{d-1} - \cdots - a_{d-1} Y - a_d \in R[Y].
\]
Consider the polynomial
\[
\phi(Y) \coloneqq Y^d \cdot f\left(\frac{1}{Y}\right) = 1 - a_1 Y - \cdots - a_{d-1} Y^{d-1} - a_d Y^d,
\]
and the formal power series 
\[
g(Y) \coloneqq \sum_{i = 0}^{\infty} \gamma_i Y^i \in R[[Y]].
\]
The product 
\[
\phi(Y) g(Y) = \sum_{i = 0}^{\infty} \gamma_i Y^i - a_1 \sum_{i = 1}^{\infty} \gamma_{i-1} Y^{i} - \cdots - a_d \sum_{i = d}^{\infty} \gamma_{i-d} Y^{i}
\]
is equal to a polynomial $h(Y) \in R[Y]$, since all its coefficients of degree $\geq d$ is zero.
Therefore $g(Y) = \frac{h(Y)}{\phi(Y)}$.

\section{Skolem Problem in Rings of Positive Characteristic}

In this section we prove Theorem~\ref{thm:Skolem}.
For a linear recurrence sequence $\gamma = (\gamma_0, \gamma_1, \gamma_2, \ldots) \in \mR^{\N}$, denote by $\mZ(\gamma)$ its zero set $\{n \in \N \mid \gamma_n = 0\}$.
We now fix a ring $\mR$ of characteristic $T > 0$ and a linear recurrence sequence $\gamma$, and Theorem~\ref{thm:Skolem} boils down to deciding whether $\mZ(\gamma)$ is empty.

Let $T = p_1^{e_1} \cdots p_k^{e_k}$ be the prime factorization of $T$.
Since $p_1^{e_1}\mR$, $\ldots$, $p_k^{e_k}\mR$ are pairwise coprime ideals of $\mR$ (i.e., $1 \in p_i^{e_i}\mR + p_j^{e_j}\mR$ for $i \neq j$), we have 
\begin{equation}\label{eq:CRT}
p_1^{e_1}\mR \cap \cdots \cap p_k^{e_k}\mR = \{0\}
\end{equation}
(c.f.~\cite[Theorem~1.3]{matsumura1989commutative}).
For each $i = 1, \ldots, k$, consider the quotient $\mA_i \coloneqq \mR/p_i^{e_i}\mR$ and the projection map $\rho_i \colon \mR \rightarrow \mR/p_i^{e_i}\mR$.
Then $\alpha_i \coloneqq \rho_i(\gamma)$ is a linear recurrence sequence over the ring $\mA_i$.
By Equation~\eqref{eq:CRT}, we have $r = 0$ if and only if $\rho_i(r) = 0$ for all $i = 1, \ldots, k$.
Therefore
\[
\mZ(\gamma) = \mZ(\alpha_1) \cap \cdots \cap \mZ(\alpha_k).
\]
We will first show that each $\mZ(\alpha_i)$ is effectively $p_i$-normal for $i = 1, \ldots, k$, and then show that their intersection is effectively computable.

\subsection{Effective $p_i$-Normality of $\mZ(\alpha_i)$}
Fix an index $i \in \{1, \ldots, k\}$. For brevity, we now omit the subscript $i$, so $\alpha$ is a linear recurrence sequence over $\mA$, where $p^e \mA = 0$ for a prime number $p$ and positive integer $e$.
In this subsection we will prove:
\propSkolemprimepower*

Let
\[
f(X) \coloneqq X^d - a_1 X^{d-1} - \cdots - a_{d-1} X - a_d \in \mA[X]
\]
be the characteristic polynomial of $\alpha$.
Observe that by extending the ring $\mA$, one can without loss of generality suppose $f$ splits in $\mA$:

\begin{lem}[folklore]\label{lem:fsplit}
    One can effectively compute a ring $\tmA \supseteq \mA$, finitely generated as an $\mA$-module, such that $f(X) = \prod_{i = 1}^t(X-r_i)^{d_i}$ with $t \in \N, r_1, \ldots, r_t \in \tmA$.
\end{lem}
\begin{proof}
    Let $T$ be a new variable and consider the quotient ring $\mA[T]/\gen{f(T)}$, where $\gen{f(T)}$ denotes the ideal generated by $f(T)$.
    The ring $\mA[T]/\gen{f(T)}$ is finitely generated as an $\mA$-module by $\{1, T, \ldots, T^{d-1}\}$.
    Since $f$ is monic, the map $\mA \rightarrow \mA[T]/\gen{f(T)}$ is injective.
    Over the base ring $\mA[T]/\gen{f(T)}$, the polynomial $f(X)$ has a linear factor, namely $X-T$. We can then replace $f(X)$ by $\frac{f(X)}{X - T}$ and obtain the lemma by induction on the degree of $f$.
\end{proof}

From now on we suppose that $f$ splits as a product $\prod_{i = 1}^t(X-r_i)^{d_i}$ over $\mA$.
When $\mA$ is a \emph{field} or an \emph{integral domain}, it is a standard result~\cite{ouaknine2015linear} that the $n$-th term of the sequence $\alpha$ can be written as an exponential polynomial $\alpha_n = \sum_{i = 1}^t r_i^{n} g_i(n)$ with $g_i(X) \in \mA[X]$, (for example by considering the Jordan normal form of the transition matrix of $\alpha$).
However this no longer holds in the presence of zero-divisors in the ring $\mA$, since these zero-divisors can neither be inverted nor localized.
To overcome this, we will use primary decomposition to reduce to the case where all zero-divisors are nilpotent.

Let $\gen{0} = P_1 \cap \cdots \cap P_s$ be the primary decomposition of the zero ideal $\gen{0} \subseteq \mA$.
Consider the projection maps $\tau_j \colon \mA \rightarrow \mA/P_j$.
Since $\gen{0} = P_1 \cap \cdots \cap P_s$, an element $a \in \mA$ is zero if and only if $\tau_j(a) = 0$ for all $j = 1, \ldots, s$.
Therefore
\[
\mZ(\alpha) = \mZ(\tau_1(\alpha)) \cap \cdots \cap \mZ(\tau_s(\alpha)).
\]
Since the intersection of $p$-normal subsets is effectively $p$-normal (Lemma~\ref{lem:intersection_normal_N}), to show $\mZ(\alpha)$ is $p$-normal it suffices to show that each $\mZ(\tau_j(\alpha)), j = 1, \ldots, s$ is $p$-normal.
Thus, we are reduced to considering the case where $\mA$ is of the form $A/P$, where $A$ is a ring with $p^eA = 0$, and $P$ is a primary ideal of $A$.
In particular, every zero-divisor of $\mA = A/P$ is nilpotent (Definition~\ref{def:commalg}(iv)).
For $n \in \Z, k \in \N$, let $\binom{n}{k}$ denote the binomial coefficient.

\begin{lem}\label{lem:exppolysum}
    Suppose every zero-divisor of $\mA$ is nilpotent, and the characteristic polynomial $f$ of $(\alpha_n)_{n \in \N}$ splits as a product $\prod_{i = 1}^t(X-r_i)^{d_i}$ with $t \in \N, r_1, \ldots, r_t \in \mA$.
    Then there is a finite, effectively computable set of non zero-divisors $S \subseteq \mA$, numbers $u, N \in \N$, a subset $\mI \subseteq \{1, \ldots, t\}$, and $c_{ik} \in S^{-1}\mA, i \in \mI; k = 1, \ldots, u$, such that
    \begin{equation}\label{eq:exppolysum}
    \alpha_n = \sum_{i \in \mI} r_i^{n} \sum_{k = 0}^{u} \binom{n}{k} c_{ik}
    \end{equation}
    for all $n \geq N$. Furthermore, every $r_i, i \in \mI$ is invertible in $S^{-1}\mA$.
\end{lem}
\begin{proof}
    Let $\mI$ be the set $\{i \mid r_i \text{ is not a zero-divisor}\}$, and let
    \[
    S = \{r_i - r_j \mid r_i - r_j \text{ is not a zero-divisor}\} \cup \{r_i \mid i \in \mI\}.
    \]
    Then every $r_i, i \in \mI$ is invertible in $S^{-1}\mA$.
    Denote $\tmA \coloneqq S^{-1} \mA$.
    Consider the formal power series ring $\tmA[[Y]]$ and let $g(Y) \coloneqq \sum_{i = 0}^{\infty} \gamma_i Y^i$.
    Note that elements of $\tmA[[Y]]$ whose constant term is one are invertible.
    Denote
    \[
    \phi(Y) \coloneqq Y^d f\left(\frac{1}{Y}\right) = \prod_{i = 1}^t(1-r_i Y)^{d_i} = 1 - a_1 Y - \cdots - a_d Y^d.
    \]
    Then $g(Y) = \frac{h(Y)}{\phi(Y)}$ for some polynomial $h(Y) \in \tmA[Y]$ (see Section~\ref{sec:prelim}).

    We claim that $\frac{1}{\phi(Y)} = \frac{1}{\prod_{i = 1}^t(1-r_i Y)^{d_i}}$ is effectively equal to a sum
    \begin{equation}\label{eq:sum_elementary}
    \sum_{i = 1}^t \sum_{j = 1}^{v} a_{ij} \cdot \frac{1}{(1-r_i Y)^{j}}
    \end{equation}
    for some $v \in \N$, $a_{ij} \in \tmA[Y], 1 \leq i \leq t, 0 \leq j \leq v$.
    The proof is done in two steps.

    \textbf{Step 1:} we reduce to the case where all $r_i - r_j, 1 \leq i < j \leq t$ are zero-divisors.
    Suppose there are $r_i, r_j$ such that $r_i - r_j$ is not a zero-divisor.
    Without loss of generality suppose these are $r_1, r_2$.
    Then $r_1 - r_2 \in S$, so $r_1 - r_2$ is invertible in $\tmA = S^{-1} \mA$.
    We have
    \[
    \frac{1}{(1-r_1Y)(1-r_2Y)} = \frac{r_1}{r_1 - r_2} \cdot \frac{1}{1 - r_1Y} + \frac{- r_2}{r_1 - r_2} \cdot \frac{1}{1 - r_2 Y}.
    \]
    Divide both sides by $(1-r_1 Y)^{d_1 - 1} (1-r_2 Y)^{d_2 - 1} \prod_{i = 3}^t(1-r_i Y)^{d_i}$, we obtain that $\frac{1}{\phi(Y)} = \frac{1}{\prod_{i = 1}^t(1-r_i Y)^{d_i}}$ is equal to a sum of two fractions
    \begin{multline}\label{eq:sum_fractions}
    \frac{r_1}{r_1 - r_2} \cdot \frac{1}{(1 - r_2 Y)^{d_2-1}\prod_{i \neq 2}(1-r_i Y)^{d_i}} + \frac{- r_2}{r_1 - r_2} \cdot \frac{1}{(1 - r_1 Y)^{d_1-1}\prod_{i \neq 1}(1-r_i Y)^{d_i}}.
    \end{multline}
    We can repeat this argument for each of the two fractions in~\eqref{eq:sum_fractions}.
    Note that the two denominators in~\eqref{eq:sum_fractions} have smaller degree than $\prod_{i = 1}^t(1-r_i Y)^{d_i}$.
    Therefore, we can apply this argument iteratively until each denominator is of the form $\prod_{i \in J}(1-r_i Y)^{d'_i}$ for some $J \subseteq \{1, \ldots, t\}$, where all $r_i - r_j, i, j \in J$ are zero-divisors.
    This reduces to the next case.

    \textbf{Step 2:} we now consider the case where all $r_i - r_j, 1 \leq i < j \leq t$ are zero-divisors.
    We will show that $\frac{1}{\phi(Y)}$ can be written in the form~\eqref{eq:sum_elementary}, using induction on $t$.
    The base case $t = 1$ is trivial.
    For the induction step, suppose $t \geq 2$.
    Since $r_1 - r_2$ is a zero-divisor and hence nilpotent, there exists $\ell \in \N$ such that $(r_1 - r_2)^{\ell} = 0$.
    Then
    \begin{align*}
    \frac{1}{1-r_1Y} & = \frac{1}{1-r_2Y} \cdot \frac{1}{\frac{1-r_1Y}{1-r_2Y}} \\
    & = \frac{1}{1-r_2Y} \cdot \frac{1}{1 - \frac{(r_1 - r_2) Y}{1-r_2Y}} \\
    & = \frac{1}{1-r_2Y} \cdot \left(1 + \frac{(r_1 - r_2)Y}{1-r_2Y} + \cdots + \frac{(r_1 - r_2)^{\ell-1}Y^{\ell-1}}{(1-r_2Y)^{\ell-1}} \right) \\
    & = \frac{1}{1-r_2Y} + \frac{(r_1 - r_2)Y}{(1-r_2Y)^2} + \cdots + \frac{(r_1 - r_2)^{\ell-1}Y^{\ell-1}}{(1-r_2Y)^{\ell}}.
    \end{align*}
    Divide both sides by $(1-r_1 Y)^{d_1 - 1} (1-r_2 Y)^{d_2} \prod_{i = 3}^t(1-r_i Y)^{d_i}$, we obtain that $\frac{1}{\phi(Y)} = \frac{1}{\prod_{i = 1}^t(1-r_i Y)^{d_i}}$ is equal to a sum of $\ell$ fractions
    \begin{equation}\label{eq:sum_fractions_l}
    \frac{1}{(1-r_1Y)^{d_1 - 1}(1-r_2Y)^{d_2 + 1}\prod_{i = 3}^t(1-r_i Y)^{d_i}} + \cdots + \frac{(r_1 - r_2)^{\ell-1}Y^{\ell-1}}{(1-r_1Y)^{d_1 - 1}(1-r_2Y)^{d_2 + \ell}\prod_{i = 3}^t(1-r_i Y)^{d_i}}.
    \end{equation}
    We can repeat this argument for each of the $\ell$ fractions in~\eqref{eq:sum_fractions_l}.
    Note that the multiplicity of the factor $1-r_1Y$ in the denominator of each fraction is smaller than $d_1$.
    Therefore, we can apply this argument iteratively until this multiplicity becomes zero.
    This eliminates the factor $1-r_1Y$, thereby decreasing $t$ by one.
    We thereby conclude the induction step.

    Combining the two steps, we have obtained
    \[
    \frac{1}{\phi(Y)} = \sum_{i = 1}^t \sum_{j = 1}^{v} a_{ij} \cdot \frac{1}{(1-r_i Y)^{j}}
    \]
    for some $v \in \N$.
    Consider a term $a_{ij} \cdot \frac{1}{(1-r_i Y)^{j}}$, we have
    \[
    \frac{1}{(1-r_i Y)^{j}} = \sum_{n = 0}^{\infty} \binom{j + n - 1}{j - 1} r_i^n Y^{n}.
    \]
    Write $a_{ij} = \sum_{s = 0}^{D_{ij}} b_s Y^s$.
    If $r_i$ is not a zero-divisor, it is invertible in $\tmA = S^{-1}\mA$, so
    \[
    a_{ij} \cdot \frac{1}{(1-r_i Y)^{j}} = \sum_{n = 0}^{\infty} \sum_{s = 0}^{D_{ij}} \binom{j + n - s - 1}{j - 1} (b_s r_i^{-s}) r_i^{n} Y^{n}.
    \]
    Its $n$-th coefficient can be written as some sum $r_i^{n} \sum_{k = 0}^{u} \binom{n}{k} c_{ik}$, by applying the identity $\binom{n+z}{j-1} = \sum_{k = 0}^{j-1} \binom{z}{j-1-k} \binom{n}{k}$.
    If $r_i$ is a zero-divisor and hence nilpotent, then $r_i^{\ell_{i}} = 0$ for some effectively computable $\ell_i \in \N$, so $a_{ij} \cdot \frac{1}{(1-r_i Y)^{j}}$ is a polynomial of degree at most $\ell_{i} + D_{ij}$.    
    This allows us to conclude the proof of the lemma by taking $N \coloneqq 1 + \max_{i \not\in \mI, 0 \leq j \leq v}\{\ell_{i} + D_{ij}\}$.
\end{proof}

From now on, we can without loss of generality replace $\mA$ with its extension $\tmA = S^{-1} \mA$ from in Lemma~\ref{lem:exppolysum}, and write 
\[
\alpha_n = \sum_{i \in \mI} r_i^{n} \sum_{k = 0}^{u} \binom{n}{k} c_{ik}
\]
for all $n \geq N$, with $c_{ik}, r_i \in \mA$ for all $i, k$.
In particular, each $r_i$ is invertible.

A linear recurrence $\alpha$ is called \emph{simple} if $u = 0$ in the above expression.
Since $p^e = 0$ in $\mA$, it is easy to see that for all $w \geq e+k!$ we have $\binom{n}{k} = \binom{n+p^{w}}{k}$ in $\mA$.
Let $w \coloneqq e+u!$, then the linear recurrence sequence $\alpha$ can be written as a union of $p^w$ different \emph{simple} recurrences $\left(\alpha_{p^wn}\right)_{n \in \N}, \left(\alpha_{p^wn + 1}\right)_{n \in \N}, \ldots, \left(\alpha_{p^wn + (p^w-1)}\right)_{n \in \N}$.
In particular, we have
\[
\alpha_{p^wn + q} = \sum_{i \in \mI} r_i^{p^wn + q} \sum_{k = 0}^{u} \binom{p^en + q}{k} c_{ik} = \sum_{i \in \mI} \left(r_i^{p^w}\right)^n d_{i,q},
\]
with $d_{i,q} \coloneqq r_i^q \sum_{k = 0}^{u} \binom{q}{k} c_{ik}$.
Since $p$-normal sets are closed under finite union and affine transformations, it suffices to show that each simple recurrence $\left(\alpha_{p^wn + q}\right)_{n \in \N}$ has an effectively $p$-normal zero set.

Therefore we now without loss of generality replace each $r_i$ with $r_i^{p^w}$ and consider each simple recurrence.
Let
\begin{equation}\label{eq:expsum}
    \alpha_n = \sum_{i \in \mI} c_i r_i^n
\end{equation}
be a simple recurrence sequence.
Since each $r_i$ is invertible, we can extend the definition of the sequence~\eqref{eq:expsum} from $n \geq N$ to all $n \in \Z$.
We now show that the solution set of $\sum_{i \in \mI} c_i r_i^z = 0$ is effectively $p$-normal in $\Z$.

\begin{lem}\label{lem:zero_set_Z}
    Let $\mA$ be a ring where $p^e = 0$ and $\mI$ be a finite set of indices.
    Let $c_i, r_i \in \mA$ for all $i \in \mI$, such that each $r_i$ is invertible.
    Then, the set of integers $z \in \Z$ such that $\sum_{i \in \mI} c_i r_i^z = 0$, is an effectively $p$-normal subset of $\Z$.
\end{lem}
\begin{proof}
    Write $\mI$ as $\{1, \ldots, t\}$ for some $t \in \N$.
    We can give $\mA$ a $\Z[X_1^{\pm}, \ldots, X_t^{\pm}]$-module structure by letting $X_i$ act as $r_i$.
    Denote by $\mM$ the submodule generated by $c_1, \ldots, c_t$, we have $p^e \mM = 0$.
    Using the standard algorithm in commutative algebra (see~\cite[Theorem~2.14]{baumslag1981computable} or~\cite[Theorem~2.6]{baumslag1994algorithmic}) we can compute a finite presentation of $\mM$ as a $\Z[X_1^{\pm}, \ldots, X_t^{\pm}]$-module.
    Then, the equation $\sum_{i=1}^t c_i r_i^z = 0$ in $\mA$ is equivalent to the equation $\sum_{i=1}^t X_i^z \cdot c_i = 0$ in the $\Z[X_1^{\pm}, \ldots, X_t^{\pm}]$-module $\mM$.

    Consider the set of solutions $(z_{11}, z_{12}, \ldots, z_{tt}) \in \Z^{t^2}$ to the S-unit equation
    \[
    X_1^{z_{11}} X_2^{z_{12}} \cdots X_t^{z_{1t}} \cdot c_1 + \cdots + X_1^{z_{t1}} X_2^{z_{t2}} \cdots X_t^{z_{tt}} \cdot c_t = 0,
    \]
    in $\mM$.
    By Theorem~\ref{thm:primepower}, the solution set $S$ is effectively $p$-normal in $\Z^{t^2}$.
    Let $\bb = (b_{11}, b_{12}, \ldots, b_{tt}) \in \Z^{t^2}$ be the vector defined by $b_{ij} = 1$ for $i = j$, and $b_{ij} = 0$ for $i \neq j$.
    Applying Corollary~\ref{cor:normalZn_to_Z} to the solution set $S$ and the vector $\bb$, we obtain that the set of integers $z$ satisfying
    \[
    X_1^{z} \cdot c_1 + \cdots + X_{t}^{z} \cdot c_t = 0,
    \]
    is effectively $p$-normal in $\Z$.
    This concludes the proof.
\end{proof}

This allows us to conclude the proof of Proposition~\ref{prop:Skolem_primepower}:

\begin{proof}[Proof of Proposition~\ref{prop:Skolem_primepower}]
    The previous discussion shows that we can reduce to the case where $\alpha_n = \sum_{i=1}^t c_i r_i^n$ for all $n \geq N$, where each $r_i$ is invertible.
    For each $0 \leq n < N$, we can verify whether $\alpha_n = 0$ individually.
    Lemma~\ref{lem:zero_set_Z} showed that the zero set
    \[
    S = \left\{z \in \Z \;\middle|\; \sum_{i=1}^t c_i r_i^z = 0 \right\}
    \]
    is effectively $p$-normal in $\Z$.
    Since $\mZ(\alpha) = \left(S \cap [N, \infty)\right) \cup F$, where $F = \{0 \leq n < N \mid \alpha_n = 0\}$, we conclude that $\mZ(\alpha)$ is an effectively $p$-normal subset of $\N$.
\end{proof}

\subsection{Intersecting $p_i$-Normal Sets}
In this subsection we prove the following proposition, which will yield Theorem~\ref{thm:Skolem} when combined with Proposition~\ref{prop:Skolem_primepower}.
\propintersectionisunion*

We start with a deep result of Karimov, Luca, Nieuwveld, Ouaknine, and Worrell:
\begin{lem}[{\cite[Theorem~3.2]{karimov2025decidability}}]\label{lem:KLNplus}
    Let $p, q \in \Z_{>0}$ be multiplicatively independent, and let $a_1, \ldots, a_k, b_1, \ldots, b_m, d \in \Z$.
    Denote by $\mS$ the set of all $(n_1, \ldots, n_k, n'_1, \ldots, n'_m) \in \N^{k+m}$ to the equation
    \begin{equation}\label{eq:KLNplus_sump_sumq}
    p^{n_1} a_1 + \cdots + p^{n_k} a_k + q^{n'_1} b_1 + \cdots + q^{n'_m} b_m = d,
    \end{equation}
    Then the set $\mS$ is of the form
    \begin{equation}\label{eq:Sform}
    \bigcup_{i \in I} \bigcap_{j \in J_i} \mX_j,
    \end{equation}
    where $I$ and each $J_i, i \in I$ is finite, and each set $\mX_j$ is in of one of the following forms
    \begin{enumerate}[nosep, label = (\roman*)]
        \item $\mX_j = \{(n_1, \ldots, n'_m) \in \N^{k+m} \mid n_s = n_t + c\}$,
        \item $\mX_j = \{(n_1, \ldots, n'_m) \in \N^{k+m} \mid n'_s = n'_t + c\}$,
        \item $\mX_j = \{(n_1, \ldots, n'_m) \in \N^{k+m} \mid n_s = c\}$,
        \item $\mX_j = \{(n_1, \ldots, n'_m) \in \N^{k+m} \mid n'_s = c\}$,
    \end{enumerate}
    with integer $c$. Furthermore, a representation of $\mS$ in the form~\eqref{eq:KLNplus_sump_sumq} can be effectively computed.
\end{lem}

From Lemma~\ref{lem:KLNplus}, we obtain the following:
\begin{lem}\label{lem:boundpq}
    Let $p, q \in \Z_{>0}$ be multiplicatively independent integers, and let $a_1, \ldots, a_k$, $b_1, \ldots, b_m$, $d \in \Q$.
    Then, there exists an effectively computable number $C \in \N$, such that all the solutions $(n_1, \ldots, n_k, n'_1, \ldots, n'_m) \in \N^{k+m}$ to the equation
    \begin{equation}\label{eq:sump_sumq}
    p^{n_1} a_1 + \cdots + p^{n_k} a_k + q^{n'_1} b_1 + \cdots + q^{n'_m} b_m = d,
    \end{equation}
    satisfy
    \begin{equation}\label{eq:boundC}
    \left| p^{n_1} a_1 + \cdots + p^{n_k} a_k \right| \leq C.
    \end{equation}
\end{lem}
\begin{proof}
    By multiplying $a_1, \ldots, a_k$, $b_1, \ldots, b_m$, $d$ with their common denominator, we can without loss of generality suppose them to be integers.
    We use induction on $k + m$.
    When $k = 0$ or $m = 0$, the bound~\eqref{eq:boundC} holds for $C = d$.
    For the induction step, suppose $k \geq 1$ and $m \geq 1$.
    By Lemma~\ref{lem:KLNplus}, we can effectively represent the solution set of Equation~\eqref{eq:sump_sumq} in the form~\eqref{eq:Sform}.
    For each $i \in I$, let $S_i \coloneqq \bigcap_{j \in J_i} \mX_j$.
    It suffices to compute a bound $C_i$ for each $i \in I$, such that $\left| p^{n_1} a_1 + \cdots + p^{n_k} a_k \right| \leq C_i$ for all $(n_1, \ldots, n_m') \in S_i$.
    Fix an index $i \in I$.
    If $S_i = \N^{k+m}$ then $a_1, \ldots, a_k$, $b_1, \ldots, b_m$, $d$ must all be zero, and the lemma becomes trivial.
    Therefore suppose $S_i \neq \N^{k+m}$. 
    Take any $\mX_j, j \in J_i$, then $\mX_j$ is in of one of the four forms in Lemma~\ref{lem:KLNplus}.
    In case~(i), every solution in $S_i = \bigcap_{j \in J_i} \mX_j$ satisfies $n_s = n_t + c$, so we can replace $n_s$ with $n_t + c$ in Equation~\eqref{eq:sump_sumq}.
    Thereby, we have replaced $p^{n_s} a_s + p^{n_t} a_t$ with $p^{n_t}(p^c a_s + a_t)$ in Equation~\eqref{eq:sump_sumq}, thereby eliminating the variable $n_s$ and decreasing $k$ by one.
    By the induction hypothesis, we can compute an effective bound~\eqref{eq:boundC} for the solutions of this new equation.
    In cases~(ii), (iii), (iv), a bound can be obtained similarly.
\end{proof}

Lemma~\ref{lem:boundpq} yields the special case of Proposition~\ref{prop:intersection_is_union} with $k = 2$:
\begin{lem}\label{lem:intersection_two_normal}
    Let $p_1, p_2$ be multiplicatively independent.
    Let $S_1, S_2$ be respectively $p_1$-normal and $p_2$-normal subsets of $\N$.
    Then $S_1 \cap S_2$ is effectively equal to a union $T_1 \cup T_2$, where $T_1$ and $T_2$ are respectively $p_1$-normal and $p_2$-normal subsets of $\N$.
\end{lem}
\begin{proof}
    For $i = 1, 2$, we have $S_i = \left(S'_i \cap [n_i, \infty)\right) \cup F_i$ where $S_1', S_2'$ are respectively $p_1$-normal and $p_2$-normal subsets of $\Z$, and $F_1 \subseteq \{0, 1, \ldots, n_1-1\}, F_2 \subseteq \{0, 1, \ldots, n_2-1\}$.
    By enlarging $F_1, F_2$ we can without loss of generality suppose $S_i = \left(S'_i \cap [n, \infty)\right) \cup F_i, i = 1, 2,$ with $n = \max\{n_1, n_2\}$ and $F_1, F_2 \subseteq \{0, 1, \ldots, n-1\}$.
    We now show that $S'_1 \cap S'_2$ is effectively a union of $p_1$-normal and $p_2$-normal subsets of $\Z$.
    
    By Observation~\ref{obs:normal_Z}, each $S'_i$ is a finite union of elementary $p_i$-nested sets of $\Z$ and sets of the form $a \Z + b$.
    Therefore it suffices to consider the intersection $U_1 \cap U_2$, where $U_i, i = 1, 2$, is elementary $p_i$-nested or $a \Z + b$.
    We show that $U_1 \cap U_2$ is either effectively $p_1$-normal or effectively $p_2$-normal.
    There are four cases.

    \textbf{Case 1:} $U_1$ is elementary $p_1$-nested and $U_2$ is elementary $p_2$-nested.
    Write
    \[
    U_1 = \{a_0 + p_1^{\ell_1 k_1} a_1 + \cdots + p_1^{\ell_1 k_r} a_r \mid k_1, k_2, \ldots, k_r \in \N\},
    \]
    and
    \[
    U_2 = \{b_0 + p_2^{\ell_2 k'_1} b_1 + \cdots + p_2^{\ell_2 k'_s} b_s \mid k'_1, k'_2, \ldots, k'_s \in \N\}.
    \]
    Take $p \coloneqq p_1^{\ell_1}$ and $q \coloneqq p_2^{\ell_2}$, they are multiplicatively independent.
    Any $z \in U_1 \cap U_2$ satisfies
    \[
    z = a_0 + p^{k_1} a_1 + \cdots + p^{k_r} a_r = b_0 + q^{k'_1} b_1 + \cdots + q^{k'_s} b_s.
    \]
    Therefore by Lemma~\ref{lem:boundpq}, we can effectively compute a bound $C \in \N$ such that every $z \in U_1 \cap U_2$ satisfies $|z| \leq C$.
    Therefore, $U_1 \cap U_2$ is effectively bounded (so it is a finite union of singletons), and is therefore both effectively $p_1$-normal and $p_2$-normal.

    \textbf{Case 2:} $U_1$ is elementary $p_1$-nested and $U_2$ is $a \Z + b$.
    In this case, both $U_1$ and $U_2$ are $p_1$-normal, therefore their intersection is effectively $p_1$-normal (Lemma~\ref{lem:intersection_normal_N}).

    \textbf{Case 3:} $U_1$ is $a \Z + b$ and $U_2$ elementary $p_2$-nested.
    Similar to the previous case, $U_1 \cap U_2$ is effectively $p_2$-normal.

    \textbf{Case 4:} $U_1 = a \Z + b$ and $U_2 = a' \Z + b'$.
    In this case $U_1 \cap U_2$ is either empty or of the form $A \Z + B$, where $A$ is the least common multiplier of $a$ and $a'$.
    Therefore $U_1 \cap U_2$ is both $p_1$-normal and $p_2$-normal.

    In all four cases, the intersection $U_1 \cap U_2$ is either effectively $p_1$-normal or effectively $p_2$-normal.
    Note that a finite union of $p$-normal sets is $p$-normal.
    Since $S'_1 \cap S'_2$ is a finite union of sets of the form $U_1 \cap U_2$, it is effectively equal to a union $T_1' \cup T_2'$, where $T_1' \subseteq \Z$ is effectively $p_1$-normal and $T_2' \subseteq \Z$ is effectively $p_2$-normal.
    Since
    \begin{multline*}
    S_1 \cap S_2 = \left(S'_1 \cap S'_2 \cap [n, \infty)\right) \cup (F_1 \cap F_2) = \left(T_1' \cap [n, \infty)\right) \cup \left(T_2' \cap [n, \infty)\right) \cup (F_1 \cap F_2),
    \end{multline*}
    we conclude that $S_1 \cap S_2$ is effectively a union of the $p_1$-normal subset $T_1 \coloneqq T_1' \cap [n, \infty) \subseteq \N$ and the $p_2$-normal subset $T_2 \coloneqq \left(T_2' \cap [n, \infty)\right) \cup (F_1 \cap F_2) \subseteq \N$.
\end{proof}

Lemma~\ref{lem:intersection_two_normal} allows us to prove Proposition~\ref{prop:intersection_is_union} in full generality:

\begin{proof}[Proof of Proposition~\ref{prop:intersection_is_union}]
    We use induction on $k$.
    The base case $k = 2$ follows directly from Lemma~\ref{lem:intersection_two_normal}.
    For the induction step, for each $i = 1, \ldots, k$, let $S_i$ be a $p_i$-normal subset of $\N$.
    By Lemma~\ref{lem:intersection_two_normal}, we can effectively write $S_1 \cap S_2 = T_1 \cup T_2$, where $T_1, T_2$ are respectively $p_1$-normal and $p_2$-normal sets.
    Then 
    \[
    S_1 \cap S_2 \cap S_3 \cap \cdots \cap S_k = (T_1 \cap S_3 \cap \cdots \cap S_k) \cup (T_2 \cap S_3 \cap \cdots \cap S_k).
    \]
    By the induction hypothesis on $k$, both $T_1 \cap S_3 \cap \cdots \cap S_k$ and $T_2 \cap S_3 \cap \cdots \cap S_k$ can be effectively written as a union of $p_i$-normal sets with $i = 1, \ldots, k$.
    Therefore $S_1 \cap S_2 \cap S_3 \cap \cdots \cap S_k$ is also effectively a union $T_1 \cup \cdots \cup T_k$ of $p_i$-normal sets with $i = 1, \ldots, k$.
    In particular, it is empty if and only if each $T_i$ is empty.
\end{proof}

Combining Proposition~\ref{prop:Skolem_primepower} and~\ref{prop:intersection_is_union}, we immediately obtain Theorem~\ref{thm:Skolem}:

\begin{proof}[Proof of Theorem~\ref{thm:Skolem}]
    Let $T = p_1^{e_1} \cdots p_k^{e_k}$ be the prime factorization of $T$, then
    $
    \mZ(\gamma) = \mZ(\alpha_1) \cap \cdots \cap \mZ(\alpha_k),
    $
    where $\alpha_i$ is the linear recurrence sequence $(\gamma \mod p_i^{e_i})$ over the quotient ring $\mR/p_i^{e_i}\mR$.
    By Proposition~\ref{prop:Skolem_primepower}, the zero set $\mZ(\alpha_i)$ is effectively $p_i$-normal.
    Then by Proposition~\ref{prop:intersection_is_union}, the intersection $\mZ(\alpha_1) \cap \cdots \cap \mZ(\alpha_k)$ is effectively equal to a union of $p_i$-normal subsets of $\N$ with $i = 1, \ldots, k$, whose emptiness can be tested.
    This yields a procedure to test whether $\mZ(\gamma)$ is empty.
\end{proof}

It is worth noting that the proof of Theorem~\ref{thm:Skolem} also shows that the zero set $\mZ(\gamma)$ is effectively a union of $p_i$-normal sets with $i = 1, \ldots, k$.

\section*{Acknowledgements} The authors would like to thank James Worrell for discussion about the Skolem Problem and S-unit equations. Ruiwen Dong is supported by a Fellowship by Examination at Magdalen College, Oxford.

\bibliography{torsion}

\end{document}